\documentclass[11pt, twoside, letterpaper]{amsart}
\usepackage{amsmath, hyperref, color, amssymb}
\usepackage{enumerate}   
\usepackage{srcltx}
\usepackage{soul}
 
\makeatletter
\@namedef{subjclassname@2020}{%
  \textup{2020} Mathematics Subject Classification}
\makeatother

\newtheorem{thm}{Theorem}[section]

\newtheorem{lem}[thm]{Lemma}
\newtheorem{prop}[thm]{Proposition}

\theoremstyle{definition}
\newtheorem{defn}[thm]{Definition}

\newtheorem{ass}[thm]{Assumption}

\theoremstyle{remark}
\newtheorem{rem}[thm]{Remark}
\newtheorem{exa}[thm]{Example}
\numberwithin{equation}{section}

\newcommand{\cB}{\mathcal{B}}

\newcommand{\cD}{\mathcal{D}}
\newcommand{\cE}{\mathcal{E}}
\newcommand{\cF}{\mathcal{F}}

\newcommand{\cM}{\mathcal{M}}

\newcommand{\cP}{\mathcal{P}}

\newcommand{\cV}{\mathcal{V}}

\newcommand{\cX}{\mathcal{X}}
\newcommand{\cY}{\mathcal{Y}}

\newcommand{\bE}{\mathbb{E}}

\newcommand{\bL}{\mathbb{L}}

\newcommand{\bN}{\mathbb{N}}

\newcommand{\bP}{\mathbb{P}}
\newcommand{\bQ}{\mathbb{Q}}
\newcommand{\bR}{\mathbb{R}}

\newcommand{\e}{\varepsilon}

\newcommand{\Pas}{{\mathbb{P}\text{-a.s.}}}

\newcommand{\dbracc}[1]{[\kern-0.15em[ #1 ]\kern-0.15em]}
\newcommand{\dbraco}[1]{[\kern-0.15em[ #1 [\kern-0.15em[}
\newcommand{\dbraoc}[1]{]\kern-0.15em] #1 ]\kern-0.15em]}
\newcommand{\dbraoo}[1]{]\kern-0.15em] #1 [\kern-0.15em[}
\newcommand{\be}{\begin{equation}}
\newcommand{\ee}{\end{equation}}
\newcommand{\nn}{\nonumber}
\newcommand{\bs}{\begin{split}}
\newcommand{\es}{\end{split}}
\newcommand{\ba}{\begin{aligned}}
\newcommand{\ea}{\end{aligned}}

\renewcommand{\[}{\left[}
\renewcommand{\]}{\right]}
\renewcommand{\(}{\left(}
\renewcommand{\)}{\right)}

\newcommand{\tb}[1]{\textcolor{blue}{#1}}
\newcommand{\tr}[1]{\textcolor{red}{#1}}

\numberwithin{equation}{section}

\renewcommand{\[}{\left[}
\renewcommand{\]}{\right]}
\renewcommand{\(}{\left(}
\renewcommand{\)}{\right)}

\begin{document}


\title[The indifference value of the weak information]{The indifference value of the weak information}
\author{Fabrice Baudoin} 
\address{Fabrice Baudoin, Department of Mathematics, Aarhus University, Aarhus, DK, 8000}%
\email{fbaudoin@math.au.dk}%

\author{Oleksii Mostovyi}\thanks{Fabrice Baudoin was partially funded by the National Science Foundation under grant No. DMS-2247117 when most of the work was completed. Oleksii Mostovyi has been supported by 
 the National Science Foundation under grant No. DMS-1848339.
 Any opinions, findings, and conclusions or recommendations expressed in this material are those of the authors and do not necessarily reflect the views of the National Science Foundation.}
\address{Oleksii Mostovyi, University of Connecticut, Department of Mathematics, Storrs, CT 06269, US}%
\email{oleksii.mostovyi@uconn.edu}%


\date{\today}%


\maketitle

\begin{abstract} 
We propose indifference pricing to estimate the value of the weak information. Our framework allows for tractability, quantifying the amount of additional information, and permits the description of the smallness and the stability with respect to small perturbations of the weak information.  
We provide {sharp} conditions for the stability with counterexamples. 
 The results rely on a theorem of independent interest on the stability of the optimal investment problem with respect to small changes in the physical probability measure. We also investigate contingent claims that are indifference price invariant with respect to changes in weak information. We show that, in incomplete models, the class of information-invariant claims includes the replicable claims, and it can be strictly bigger. In particular, in complete models, all contingent claims are information invariant. We augment the results with examples and counterexamples.  
\end{abstract}  
 \section{Introduction}

  Asymmetry of information is a very active area of mathematical finance and related areas. Going back to \cite{Marschak59}, \cite{ArrowInf}, \cite{Gould74}, \cite{Kyle}, and \cite{BackIT}, among others, the topic resulted in numerous results on  the asymmetry of information and related subjects.
One of the mainstream approaches relies on enlargements of filtrations, where the mathematical foundations have been largely developed by the French school of probability.  
 The financial applications of this theory, known as strong information modeling in the terminology of \cite{Fabrice03}, have propagated to a range of topics in mathematical finance, including arbitrage theory, pricing and hedging, characterizations equilibria, optimal investment, and others. We refer to   \cite{KarPik}, \cite{MartinII2}, \cite{higa}, \cite{imkellerAsym},  \cite{CampiII}, 
  \cite{ImkellerAI},\cite{umut1}, \cite{KostasInformation}, \cite{Michael1}, \cite{HaoIT}, \cite{aksamit1},  \cite{umut2}, \cite{kosBeaFont16},  
 \cite{moniqueAksamit},  \cite{aksamit4}, \cite{philip1}, \cite{aksamit3},\cite{aksamit2},  \cite{PaoloLP}, \cite{philip2}, \cite{claudioIA}, \cite{ScottAI1}, \cite{umut3}, \cite{ScottAI2}, \cite{Choi2023}, and \cite{Shi2024} for an incomplete list of references on these subjects.

In the present work, we propose to use \begin{enumerate}[(i)]
\item {\bf  indifference pricing} in the context of 
\item {\bf  the weak information approach}  from \cite{Fabrice03} 
\end{enumerate}to quantify the value of information. 
For $\rm{(i)}$, in contrast, e.g., to the changes in the value function with and without extra information that is used in a \cite{MartinII}, \cite{Fabrice03}, \cite{imkellerAsym}, \cite{philip1}, \cite{claudioIA}, and \cite{philip2},  among others, using indifference pricing allows to assign different value of information between different securities. In particular, {\it there are some contingent claims that are not affected by the additional information.}

For $\rm{(ii)}$, 
in contrast to the strong information modeling, the weak information approach does not require changing the filtration but relies on alterations of the physical probability measure. It permits the recovery of many results from the {initial} enlargements of filtrations and beyond, see \cite{Fabrice02} and \cite{Fabrice03}, and allows for multiple desired features of the value of information problem. For example, smallness or even continuity with respect to changes in the information can be naturally discussed in the context of our framework. This is in contrast to the approach based on the enlargement of filtrations, see, e.g., deep yet technically involved \cite{KostasInformation}, where continuity with respect to small changes in the filtration is analyzed. Here, we note that continuous behavior with respect to small changes in the initial data is a part of the well-posedness of a problem in the sense of Hadamard.

In addition to proposing the framework via $\rm{(i)}$ and $\rm{(ii)}$, our contributions also include 
\begin{enumerate}[(i)]
\setcounter{enumi}{2}
\item {\bf sharp stability results} (with counterexamples in Section \ref{secCounter}) with respect to small perturbations of the physical measure,
\item {\bf characterizations of indifference price invariant contingent claims} under changes of the weak information.   
\end{enumerate}
The results in $\rm{(iii)}$ are of independent interest, and they are based on technical proofs of the convergence of the value function, optimizers, and indifference prices under small perturbations of the physical probability measure. Here, we introduce the integrability conditions and, under these conditions, prove convergence of the value functions, their optimizers, and indifference prices under small perturbations of physical probability measure. We augment this analysis with counterexamples.

For  $\rm{(iv)}$, we show that the class of indifference price invariant contingent claims includes all bounded claims in complete markets. In incomplete markets, we prove that the set of indifference prices invariant claims includes all bounded replicable contingent claims. Depending on a particular mode, this set can equal the set of replicable claims, or it can be strictly bigger and might include even all bounded contingent claims. We illustrate these assertions with positive examples.

   The remainder of this paper is organized as follows: in Section \ref{secModel}, we formulate the model and state the stability results, whose proofs are given in Section \ref{secProofs}. In Section \ref{secInv}, we discuss the price invariant contingent claims in both complete and incomplete markets, establish a connection to the framework in \cite{Fabrice03}, and provide positive examples. We conclude the paper with Section \ref{secCounter}, which contains counterexamples.

 \section
 {The model and the stability results}\label{secModel}

Let us consider a  probability space $(\Omega,\cF,\bP)$, equipped with the filtration $\(\cF_t\)_{t\in[0,T]}$, where $T \in(0,\infty)$ is the time horizon, $\cF$ satisfies the usual conditions, and $\cF_0$ is a trivial $\sigma$-algebra. Let $S$ be a $d$-dimensional semimartingale denoting the price process of the risky assets. We suppose that there is also a riskless asset whose price process equals $1$ at all times. Let us define the family of nonnegative self-financing wealth processes as  
\be\label{defX}\begin{split}
\cX (x) := \left\{X \geq 0 :~ \right.&X = x + H\cdot S,~where~H~is~predictable\\
&\left.and~S\text-integrable\right\},\	\quad x\geq 0.
\end{split}
\ee
We remark that 
the wealth processes of the form $X+x_0$, where $x_0\in \bR$ and $X\in\bigcup\limits_{x\geq 0}\cX(x)$, are called {\it admissible}\footnote{This notion is used, in particular, in Definition \ref{defRepl} below to describe bounded replicable contingent claims.}.

 A {\it stochastic utility field} is a mapping $U=U(\omega, x):~\Omega\times[0,\infty)\to \bR\cup\{-\infty\}$ satisfying the following assumption.
\begin{ass}\label{Inada}
For every $\omega\in\Omega$, $U(\omega, \cdot)$ is a strictly increasing, strictly concave, continuously differentiable on $(0,\infty)$ function that satisfies the Inada conditions at $0$ and $\infty$; 
for every $x\geq 0$, $U(\cdot, x)$ is measurable. At $0$, we suppose by continuity that $U(\cdot, 0) = \lim\limits_{x\downarrow 0 }U(\cdot, x)$; this value might be $-\infty$.  
\end{ass}

 \subsection{Weak anticipation}\label{secWeakAntip}  Let $\cP$ be a Polish space, e.g., $\cP = \bR^n$, $\cP = C\(\bR_{+}, \bR^n\)$, for some $n\in\bN$, etc., endowed with its Borel $\sigma$-algebra $\cB(\cP)$. Next, let $\sf Y:\Omega\to \cP$ be an $\cF$-measurable random element. In the simplest case, $\sf Y$ can be the value of the stock price at time $T$.

Let us consider an insider who is weakly informed of  $\sf Y$, that is, he or she has knowledge of the filtration $\(\cF_t\)_{t\in[0,T]}$ and of the {\it law} of $\sf Y$. So, we associate a probability measure $\nu$ on $\(\cP, \cB(\cP)\)$. We assume that $\nu$ is equivalent to $\bP$ and note that in \cite{Fabrice03}, more conditions are imposed. Following \cite[Definition 7]{Fabrice03}, one can set a probability measure $\bP^\nu$ on $(\Omega, \cF)$ via 
$$\bP^\nu(A) = \int_{\cP} \bP[A|\sf Y = y]\nu(dy),\quad A\in\cF,$$
which is called, in \cite{Fabrice03}, the minimal probability measure associated with the weak information $(\sf Y, \nu)$. 

 Multiple results for the initial enlargement of filtration can be recovered from the weak anticipation approach. These include the semimartingale decomposition in the Jacod theorem - the key component for quantifying the value of information and results on stochastic differential equations, among others. 
We refer to \cite{Fabrice02} and \cite{Fabrice03} for more details, where additionally dynamic conditioning, modeling weak information flow, and even connections to the theory of Schr\"odinger processes are developed. 
 
 \subsection{Our formulation} The weak anticipation approach allows us to consider small perturbations of information by supposing that $\bP^\nu$ is close in some sense to $\bP$. In turn, this leads to the concepts of continuity and stability of various problems with respect to small changes in information that can be modeled via alterations of the physical probability measure. 
 
 Therefore, we can consider a sequence of equivalent probability measures on one fixed probability space, converging to a physical probability measure $\bP$, and see how various problems related to indifference pricing respond. For the notion of convergence of probability measures on a fixed space, we propose the total variation norm. 
 These considerations are represented by the following assumption.
  \begin{ass}\label{asZconv}
 Let us consider a sequence of probability measures $\bP^n$, $n\in\bN$, equivalent to $\bP$, and let  $Z^n_T := \frac {d\bP^n}{d\bP}$,  $n\in\bN$,  be the sequence of the corresponding   Radon-Nikodym derivatives.
We suppose that 
$$\bP\text-\lim\limits_{n\to\infty}Z^n_T  = 1 =:Z^\infty_T.$$

\end{ass}
\begin{rem}\label{remTVconv}
Assumption \ref{asZconv} by means of Scheffe's lemma implies that 
$$\bL^1(\bP)\text-\lim\limits_{n\to\infty}Z^n_T=  1.$$
We refer to \cite[Section V.4]{JS} for characterizations of such a convergence. 
\end{rem}
Let us  denote $\bN^*:= \bN\cup{\{\infty\}}$, $\bP^\infty:=\bP$, and consider the following problems.
\be\label{primalProblem}
u_n(x) := \sup\limits_{X\in\cX(x)}\bE_{\bP^n}\[U(X_T)\],\quad (x, n)\in(0,\infty)\times \bN^*.
\ee
Here, we use the convention 
\be\label{convu}
\bE_{\bP^n}\[U(X_T)\]:= -\infty\quad {\rm if}\quad \bE_{\bP^n}\[U^-(X_T)\] = \infty.
\ee
Let us further set 
\be\label{defY}\bs
\cY^n(y): = \left\{ \right.&Y\geq 0:~ Y_0\leq y~~and~~ XY = (X_tY_t)_{t\in[0,T]}~~is~a \\
&\left.\bP^n\text-supermartingale~for~every~X\in\cX(1)\right\}, \quad (y, n)\in[0,\infty)\times \bN^*,
\end{split}
\ee
where, for every $y\geq 0$, below we denote $\cY^\infty(y)$ by $\cY(y)$ 
for brevity. 
$$V(\omega, y):=\sup\limits_{x>0}\(U(\omega, x) - xy\), \quad (\omega, y)\in \Omega\times[0,\infty),$$
and 
\be\label{dualProblem}
v_n(y):= \inf\limits_{Y\in\cY^n(y)}\bE_{\bP^n}\[V(Y_T)\], \quad (y, n)\in(0,\infty)\times \bN^*,
\ee
where we use the convention
\be\label{convv}
\bE_{\bP^n}\[ V(Y_T)\] = \infty\quad {\rm if}\quad
\bE_{\bP^n}\[V^+(Y_T)\] = \infty.
\ee
To ensure that the base dual problem, that is, the one corresponding to $n=\infty$, is non-degenerate, let us suppose that 
\be\tag{noArb}\label{noArb}
\cY(1)~{\rm contains~a~strictly~positive~element}.
\ee
This is the celebrated condition of no-unbounded profit with bounded risk (NUBPR). Further, let us impose the following assumptions.

\begin{ass}\label{asUI1}
For every $x>0$, there exists $X\in\cX(x)$, such that 
$Z^n_TU^{-}(X_T)$, $n\in\bN$, is uniformly integrable. 
\end{ass}
\begin{rem}\label{remDetU}
If $U$ is deterministic, Assumption \ref{asUI1} holds trivially under Assumption \ref{asZconv}. 
\end{rem}
\begin{ass}\label{asUI}
For every $y>0$, there exists $Y\in\cY(y)$, such that 
$Z^n_TV^{+}\(\frac{Y_T}{Z^n_T}\)$, $n\in\bN$, is uniformly integrable. 
\end{ass}
\begin{rem}
Let us consider the following condition 
\be\label{finValue}
  u_n(z) >-\infty  \quad and \quad 
  v_n(z)  <\infty, \quad z>0,\quad n\in\bN^*,
\ee 
 which  implies that for every $y>0$, there exists $Y^n\in\cY(y)$, $n\in\bN$, such that 
$Z^n_TV^{+}\(\frac{Y^n_T}{Z^n_T}\)$, $n\in\bN^*$, is a subset of $\bL^1(\bP)$. Assumption \ref{asUI} is closely related  yet stronger. 

\end{rem}

 Next, we state our convergence results. 
 \subsection{Convergence of the value functions}
\begin{prop}\label{propconvuv}
Let us consider problems \eqref{primalProblem} and \eqref{dualProblem} under Assumptions \ref{Inada}, \ref{asZconv}, \ref{asUI1}, \ref{asUI}, and \eqref{noArb}. Then, we have 
\begin{enumerate}
\item condition \eqref{finValue} holds, and for every $n\in\bN^*$,  
$u_n$ and $v_n$ are finite-valued on $(0,\infty)$. 
\item the value functions converge pointwise, that is 
\be\label{convuv}
\lim\limits_{n\to\infty}u_n(z) = u_{\infty}(z)\quad and \quad 
\lim\limits_{n\to\infty}v_n(z) =  v_{\infty}(z),\quad z>0.
\ee
\end{enumerate}
\end{prop}

\subsection{Convergence of the optimizers}
\begin{prop}\label{lemConvOpt}

Let $(z_n)_{n\in\bN}$ be a sequence of strictly positive numbers converging to $z_\infty>0$. Then, 
under the conditions of Proposition \ref{propconvuv},  for every $(z_n, n)$, $n \in\bN^*$, the optimizers to \eqref{primalProblem}, denoted by $\widehat X^{n}(z_n)$, and \eqref{dualProblem}, denoted by $\widehat Y^n(z_n)$,  exist, are unique, and satisfy
\be\label{convOpt}
 \widehat X^\infty_T(z_\infty) = \bP{\text-}\lim\limits_{n\to\infty}\widehat X^{n}_T(z_n)\quad and\quad 
 \widehat Y^\infty_T(z_\infty) = \bP{\text-}\lim\limits_{n\to\infty}\widehat Y^n_T(z_n). 
\ee
 
\end{prop}

\subsection{Convergence of the indifference prices}
Let us recall the definition of indifference prices. Fix $f \in\bL^\infty\(\Omega, \cF, \bP\)$. Then, for every $(x,q)\in\bR^2$, let us  denote
\be\label{defXxq}\begin{split}
\cX(x,q) :=\left\{ X = x + H\cdot S:~~ X\geq -x'(X)~and~X_T + qf\geq 0\right\},
\end{split}
\ee
where $x'(X)$ is some non-negative constant, which can be different from $x$ and which can depend on $X$.

\begin{defn}\label{defIndPrice} Following \cite[page 157]{KarKar21}\footnote{{This definition, in a different form, goes back to \cite{Davis97}.}}, we define the indifference price of $f$ corresponding to the initial wealth $x$, a utility stochastic field $U$, and a probability measure $\bQ$ is a constant $\Pi =  \Pi\(f,x,U, \bQ\)$, such that
\be\label{defIP}
\bE_{\bQ}\[U(X_T +qf)\]\leq u(x),\quad  q\in\bR, \quad X\in \cX(x-q\Pi,q),
\ee
where $u(x)$ is given by \eqref{primalProblem}  corresponding to the stochastic utility field $U$ and the probability measure $\bQ$, that is $u(x) = \sup\limits_{X\in\cX(x)}\bE_{\bQ}\[U(X_T)\]$, $x>0$. 
\end{defn}

To focus on probability measures and utility stochastic fields that lead to well-posed problems of computing the indifference prices, we introduce the class of utility stochastic fields that have finite both primal and dual value functions under a given probability measure $\bQ$, satisfy \eqref{finValue} for one fixed $n$, that is 
\be\label{finValueU}
 \sup\limits_{X\in\cX(x)}\bE_{\bQ}\[ U(X_T)\]>-\infty, \quad x>0,
\ee
\be\label{finValueV}
 \inf\limits_{Y\in {\cY^{\bQ}}(y)}\bE_{\bQ}\[ V(Y_T)\]<\infty,\quad y>0,
\ee
where ${\cY^{\bQ}}$'s are given by \eqref{defY}, with one fixed $\bQ$ being used instead of $\bP^n$'s. 
\begin{defn}\label{defFV}
Let a given probability measure $\bQ$ be fixed. We define the class of utility stochastic fields satisfying \eqref{finValueU} and whose conjugate satisfies \eqref{finValueV} by 
$\cF\cV(\bQ)$. 
\end{defn}
One can see that $\cF\cV(\bQ)\neq \emptyset$ for every $\bQ\sim\bP$, as $\cF\cV(\bQ)$ includes all deterministic Inada utility functions, whose convex conjugates are bounded from above.
\begin{prop}\label{propUBPconv}
Let the assumptions of Proposition \ref{propconvuv}, and let us consider a sequence of strictly positive numbers $x_n$ converging to $x_\infty>0$. Let $y_n:=u'_n(x_n), n \in\bN^*$, and suppose that 
$\widehat Y^n(y_n)$ is a $\bP^n$-martingale,  $n \in\bN^*$. Then, for every bounded contingent claim $f$, $x>0$, probability measure $\bP_n$ and $U_n\in\cF\cV(\bP_n)$, the indifference prices $\{p_n\}  = \{\Pi(f, x_n, U_n, \bP_n)\}$ are singletons, $n\in\bN^*$, and we have
\be\label{UBPconv}
\lim\limits_{n\to\infty}p_n = p_\infty,
\ee
where each $p_n$ has the representation
\be\label{ubpRep}
p_n = \bE_{\bP^n}\[\frac {\widehat Y^n_T(y_n)}{y_n}f\]=\bE \[Z^n_T\frac {\widehat Y^n_T(y_n)}{y_n}f\],\quad n \in\bN^*.
\ee
\end{prop}
\subsection{On sufficient conditions for the martingale property of the dual minimizer}
Sufficient conditions for the martingale property of the dual minimizer are obtained in the following works. 
If the returns of risky assets are modeled with Levy processes and with power deterministic utilities, the martingale property is established in \cite{KallsenLevy}; if the returns are driven by the processes with independent increments (PII) in the sense of \cite[Chapter II]{JS}, and in power utility settings, the martingale property of the dual minimizer is obtained in \cite{KallsenPII}; characterizations of the martingale property via the reverse H\"older inequality are obtained in \cite{Marcel}. Beyond the power/logarithmic utilities, \cite{MostovyiSirbuThaleia} connects the martingale property to the stochastic dominance of the dual domain of any order from second and higher up to and including the infinite order. 

%
%
%
%

%
%
%
 
\subsection{On relaxation on the boundedness condition on $f$ in Proposition \ref{propUBPconv}}
The stochastic utility settings considered in the paper, in particular, allow the relaxation of the boundedness condition on the contingent claim $f$ in Proposition \ref{propUBPconv}. While, in our view, the complete analysis of this question is a topic of a separate investigation, we consider the following example.

Let us suppose that $f = \max(S^1_T - K, 0)$, which is a payoff of a call option on a risky asset. In this case, in a wide class of models, for example, in the classical Black-Scholes model of the stock price, if the strike $K$ is deterministic, the option $f$ is replicable. This implies that it admits a unique arbitrage-free price, and its indifference price equals the arbitrage-free price (see, e.g., the argument in \cite[Lemma 6.2]{MostovyiPietroLMP}). 

However, if the dynamic of risky assets admits jumps, as in \cite{MertonJumps}, \cite{KouJumps}, \cite{Kou2}, and \cite{ContTankov}, or if the strike $K\geq 0$ is random, the payoff of the European call option is nonreplicable, see \cite[Chapter 11]{Shreve2004}. Let $S^1$ be strictly positive and maximal in $\cX(S^1_0)$, where by the maximality we mean that $S^1_T$ cannot be dominated by the terminal value of any other element of $\cX(S^1_0)$, as in the Black-Scholes model, and let us suppose, without loss of generality, that $S^1_0 = 1$. 

Next, starting from a stochastic utility field $U$ satisfying Assumption \ref{Inada}, we introduce the following auxiliary stochastic utility field
\be\label{tildeU}
\tilde U(\omega, x) : = U(\omega, S^1_T(\omega) x),\quad (\omega, x)\in \Omega\times[0, \infty),
\ee
and observe that $\tilde U$ also satisfies Assumption \ref{Inada}.
Then, we have that 
$$\tilde V(\omega, y) := V\(\omega, \frac y{S^1_T(\omega)}\),\quad (\omega, y)\in \Omega\times[0, \infty),
$$
is the (pointwise in $\omega$) convex conjugate of $\tilde U(\omega, x)$. Next, in \eqref{defXxq}, we have to replace the admissibility, that is, boundedness from below by a constant by the acceptability in the sense of \cite{Delbaen-Schachermayer1997}, where we say a wealth process $X$ is acceptable if it can be written as $X' - X''$, where $X'$ is a nonnegative wealth process, and $X''$ is maximal.  Now, for the unbounded $f$ as above, we replace  \eqref{defXxq} with 
\be\label{defXxqa}
\cX(x,q) : = \left\{acceptable~X:~X_0 = x~and~X_T + qf\geq 0\right\},~~ (x,q)\in\bR^2.
\ee
Next, for every  $X\in\cX(x,q)$, $(x,q)\in\bR^2$, such that $\cX(x,q)\neq \emptyset$, from \eqref{tildeU}, we obtain 
$$U(X_T + qf) = \tilde U\(\frac {X_T}{S^1_T} + q \frac f{S^1_T}\).$$
Setting $$\tilde \cX(x) := \left\{ \frac {X}{S^1}:~X\in\cX(x)\right\}, \quad x>0,$$ we can restate the primal optimization problem as
\be\label{primalProblem2}\begin{split}
u(x) &= \sup\limits_{X\in\cX(x)}\bE\[U(X_T)\] \\
&=  \sup\limits_{X\in\cX(x)}\bE\[\tilde U\(\frac{X_T}{S^1_T}\)\] = \sup\limits_{X\in\tilde{\cX}(x)}\bE\[\tilde U\( {X_T} \)\],\quad  x>0.\\
\end{split}
\ee
For the dual optimization problem (without random endowment), with 
$$\tilde \cY(y): = \left\{ YS^1:~Y\in\cY(y)\right\}, \quad y>0,$$ we set 
$$v(y) = \inf\limits_{Y\in\cY(y)}\bE\[ V(Y_T)\] = \inf\limits_{Y\in\cY(y)}\bE\[ \tilde V(Y_TS^1_T)\] =  \inf\limits_{Y\in\tilde \cY(y)}\bE\[ \tilde V(Y_T)\],\quad y>0.$$
 The change of num\'eraire allows us to specify the utility maximization problem (via \eqref{tildeU}) with a {\it bounded} contingent claim 
\be\label{deftildef}
\tilde f: = \frac f{S^1_T},
\ee
 as follows. 
With $\cX(x,q)$'s being defined in \eqref{defXxqa}, we set 
\be\label{defTildeX}
\tilde \cX(x,q) = \left\{\frac X{S^1}:~X\in\cX(x,q)\right\},
\ee where, by the maximality of $S^1$ in $\cX(1) = \cX(1,0)$, we note that maximal elements of $\cX(1,0)$ are exactly the maximal elements of $\tilde\cX(1,0)$ multiplied by $S^1$. Finally, for every $X\in\cX(x,q)$, $(x,q)\in\bR^2$, with $\tilde X: = \frac{X}{S^1}$, we have 
\be\nn\bs
\bE\[ U(X_T + qf)\] &= \bE\[ \tilde U\(\frac{X_T + qf}{S^1_T}\)\] =  \bE\[ \tilde U\( {\tilde X_T + q\tilde f} \)\],
\end{split}\ee
where the latter expression is given in terms of a {\it bounded} $\tilde f$. This allows us to start from unbounded contingent claim $f=\max(S^1_T - K, 0)$ by changing num\'eraire to $S^1$ and via introducing another Inada utility stochastic field in \eqref{tildeU} to reformulate the notion of an indifference price in Definition \ref{defIndPrice}, particularly  \eqref{defIP}, in terms of a {\it bounded} contingent claim $\tilde f$ given in \eqref{deftildef} and stochastic Inada utility field defined in \eqref{tildeU}. 

In more detail, extending Definition \ref{defIndPrice} to {\it unbounded} contingent claims in the context of this section, particularly to $f = \max(S^1_T - K, 0)$, we can define 
an indifference price of $f$ corresponding to the initial wealth $x$, a utility stochastic field $U$, and a probability measure $\bQ$, to be a constant $\Pi =  \Pi\(f,x,U, \bQ\)$, such that
\be\label{IP2}
\bE_{\bQ}\[U(X_T +qf)\]\leq u(x),\quad  q\in\bR, \quad X\in \cX(x-q\Pi,q),
\ee
where $u(x)$ is given by \eqref{primalProblem2}. This constant $\Pi$ is also given by $\Pi\(\tilde f,x,\tilde U, \bQ\)$ in the sense of Definition \ref{defIndPrice}, see \eqref{defIP}, that is $\Pi$ satisfies \eqref{IP2} if and only if $\Pi$ satisfies 
\be\label{IP3}
\bE_{\bQ}\[\tilde U(X_T +q\tilde f)\]\leq u(x),\quad  q\in\bR, \quad X\in \tilde \cX(x-q\Pi,q),
\ee
where $\tilde U$ is set in \eqref{tildeU}, $u(x)$ is still given by \eqref{primalProblem2}, $\tilde f$ is defined in \eqref{deftildef}, and $\tilde \cX(x,q)$'s are specified in \eqref{defTildeX}. 

We conclude that \eqref{tildeU}, \eqref{deftildef}, \eqref{defTildeX}, and \eqref{IP3} allow us to characterize the indifference price of {\it unbounded} $f$ via the results of the remainder of this paper obtained for {\it bounded} contingent claims. 
\section{Proofs}\label{secProofs}
 
The proof of Proposition \ref{propconvuv} is given via the following lemmas. 
We begin with a structural lemma that establishes a relationship between the domains of the dual problems.
 \begin{lem}
 \label{lemStructure}
Let us 
suppose that Assumption \ref{asZconv} and \eqref{noArb} hold and consider sets $\cY^n$, $n\in\bN^*$, defined in \eqref{defY}. Then, we have 
\be\bs\nn
\cY^n(y) &= \left\{\frac{Y}{Z^n} = \left(\frac{Y_t}{Z^n_t}\right)_{t\in[0,T]}:~Y\in\cY(y)\right\}\neq\emptyset \quad {and}\\
 \cY(y) &= \left\{Y^nZ^n=  \left({Y_t}{Z^n_t}\right)_{t\in[0,T]}:~Y^n\in\cY^n(y)\right\},\quad (y, n)\in(0,\infty)\times \bN^*,
\end{split}
\ee that is 
 $$Y^n\in\cY^n(y)\quad \Leftrightarrow\quad Y^n Z^n\in\cY(y), \quad (y, n)\in(0,\infty)\times \bN^*,$$
as well as 
 $$Y\in\cY(y) \quad \Leftrightarrow\quad \frac Y{Z^n}\in \cY^n(y), \quad (y, n)\in(0,\infty)\times \bN^*.$$
 \end{lem}
 \begin{proof}
 The assertions immediately follow from  \eqref{defY} and Assumption \ref{asZconv} via \cite[Lemma 3.5.2]{KaratzasShreve1}. 
 \end{proof}
 
\begin{lem}\label{lemFinValue}
Under the assumptions of Proposition \ref{propconvuv}, we have that condition \eqref{finValue} holds, and furthermore, for every $n\in\bN^*$, $u_n$ and $v_n$ are finite-valued on $(0,\infty)$. 
\end{lem}
\begin{proof}
Let us fix $n\in\bN^*$ and $z>0$. From conjugacy between $U$ and $V$, we get
$$U(x) \leq V(y) + xy,\quad (x,y)\in\bR^2,$$
and thus for $X\in\cX(z)$ and $Y\in\cY(z)$ as in Assumptions \ref{asUI1} and \ref{asUI}, respectively, using Lemma \ref{lemStructure}, we get  
\be\label{1191}
Z^n_T U(X_T)\leq 
Z^n_T V\(\frac{Y_T}{Z^n_T}\) + X_T  {Y_T},\quad \Pas.
\ee
Let us consider $Z^n_T U(X_T)$. Taking the expectation, we get
\be\label{1192}
u_n(z)\geq \bE\[Z^n_T U(X_T)\]\geq  \bE\[Z^n_T\(- U^-(X_T)\)\]>-\infty,
\ee
where, in the last inequality, we have used Assumption \ref{asUI1}. Likewise, in the right-hand side of \eqref{1191}, taking the expectation, we obtain  
\be\label{1193}
v_n(z)\leq \bE\[Z^n_T V\(\frac{Y_T}{Z^n_T}\) \]\leq \bE\[Z^n_T V^+\(\frac{Y_T}{Z^n_T}\) \]<\infty,
\ee
where the last inequality follows from Assumption \ref{asUI}.  As $z>0$ and $n\in\bN^*$ are arbitrary, \eqref{1192} and \eqref{1193} imply \eqref{finValue}. Furthermore,  one can show  that 
\be\label{1194}
u_n(z) \leq v_n(z) + z^2,\quad z>0,\quad n\in\bN^*.
\ee
Then, 
\eqref{1192} and \eqref{1193} together with \eqref{1194} imply that $u_n$ and $v_n$ are finite-valued on $(0,\infty)$. 
\end{proof}

\begin{lem}\label{lemPrimalBound}
Under the conditions of Proposition \ref{propconvuv}, we have 
\be\label{10253}
\liminf\limits_{n\to\infty}u_n(x) \geq u_\infty(x),\quad x>0.
\ee
 \end{lem}
 \begin{proof}

Let us fix $x>0$ and $\e\in (0,x)$, and consider $X\in\cX(\e)$, such that 
$Z^n_TU^{-}(X_T)$, $n\in\bN$, is uniformly integrable. The existence of such $X$ follows from Assumption \ref{asUI1}. With $X^\infty_T(x -\e)$ denoting the optimizer associated with $u_\infty$ at $(x-\e)$, we have 
\be\label{10251}
u_n(x) \geq \bE\[Z^n_TU(X_T + \widehat X^\infty_T(x -\e))\], \quad n\in\bN^*.
\ee
Next, from the monotonicity of $U$, we get 
 \be\nn
 Z^n_TU^{-}(X_T + \widehat X^\infty_T(x -\e))\leq Z^n_TU^{-}(X_T), \quad n\in\bN^*,
\ee
and therefore, using Assumption \ref{asUI1}, 
we get 
\be\label{10252}
Z^n_TU^{-}(X_T + \widehat X^\infty_T(x -\e)), ~~ n\in\bN^*, ~~ {\rm is~ uniformly~integrable~under~}\bP.
\ee
Consequently, from \eqref{10251} and \eqref{10252}, using Fatou's lemma, we get
\be\nn\bs
&\liminf\limits_{n\to\infty}u_n(x) \geq \liminf\limits_{n\to\infty}\bE\[Z^n_TU(X_T + \widehat X^\infty_T(x -\e))\] \\
&\geq \bE\[ U(X_T + \widehat X^\infty_T(x -\e))\] \geq  \bE\[ U(\widehat X^\infty_T(x -\e))\] = u_\infty(x-\e).
\end{split}
\ee
As $\e>0$ is arbitrary, by taking the limit as $\e\downarrow 0$, and using continuity of $u_\infty$ at $x$ (by the finiteness of $u_\infty$ established in Lemma \ref{lemFinValue}, and concavity), we conclude that \eqref{10253} holds. 

 \end{proof}
 \begin{lem}\label{lemDualBound}
 Under the conditions of Proposition \ref{propconvuv}, we have 
  \be\label{9222} 
\limsup\limits_{n\to\infty}v_n(y) \leq v_\infty(y), \quad y>0.
\ee
 \end{lem} 
 \begin{proof}
 Let us fix $y>0$ and $\e\in(0,y)$ and denote  $ \widehat Y^\infty(y -\e)$ the dual minimizer corresponding associated with $v_\infty$ at $(y-\e)$. Then, using Lemma \ref{lemStructure}, for every $\tilde Y\in\cY$, we get  
\be\label{10257}
v_n(y) \leq \bE\[Z^n_T V\(\frac{\e\tilde Y_T + \widehat Y_T^\infty(y -\e)}{Z^n_T}\)\],\quad n\in\bN^*.
\ee
By the monotonicity of $V$, we get
\be\label{1196}
Z^n_T V^+\(\frac{\e\tilde Y_T + \widehat Y_T^\infty(y -\e)}{Z^n_T}\)\leq Z^n_T V^+\(\frac{\e\tilde Y_T}{Z^n_T}\),\quad n\in\bN^*.
\ee
With ${Y}$ being as in Assumption \ref{asUI} for $y = \e$, let us consider $\tilde Y = \frac{Y}{\e}$. Then it follows, from \eqref{1196} and Assumption \ref{asUI}, that the sequence $Z^n_T V^+\(\frac{\e\tilde Y_T + \widehat Y_T^\infty(y -\e)}{Z^n_T}\)$, $n\in\bN^*$, is uniformly integrable. 
 Therefore, using Fatou's lemma, we obtain
\be\label{10256}\bs
&\limsup\limits_{n\to\infty}v_n(y) \leq\limsup\limits_{n\to\infty} \bE\[Z^n_T V\(\frac{\e\tilde Y_T + \widehat Y_T^\infty(y -\e)}{Z^n_T}\)\] \\ 
&\leq  \bE\[ V\( \e\tilde Y_T + \widehat Y_T^\infty(y -\e)  \)\] \leq \bE\[ V\(  \widehat Y_T^\infty(y -\e)  \)\] = v_\infty(y-\e).
\end{split}\ee
As $\e>0$ in \eqref{10256} is arbitrary, taking the limit as $\e\downarrow 0$  and using the continuity of $v_\infty$ at $y$, by its finiteness, established in Lemma \ref{lemFinValue}, and convexity, we get \eqref{9222}.

 \end{proof}

\begin{proof}[Proof of Proposition \ref{propconvuv}]
For $n = \infty$, let us fix $x>0$ and let $y = u'_{\infty}(x)$. We notice that the differentiability of $u_\infty$ follows from \cite[Theorem 3.2]{MostovyiNec} 
Then, we have 
   \be\label{9223}
 \bs
 u_\infty(x)\leq \liminf\limits_{n\to\infty}u_n(x)  \leq 
 \limsup\limits_{n\to\infty}v_n(y) + xy \leq   v_\infty(y) + xy = u_\infty(x),
 \end{split}
 \ee
 where, in the first inequality, we used Lemma \ref{lemPrimalBound}; in the second - conjugacy relations; in the third - Lemma \ref{lemDualBound}; and in the (last) equality - \cite[Theorem 3.2]{MostovyiNec}. 
As $x>0$ is arbitrary and,  by \cite[Theorem 3.2]{MostovyiNec},  $u'_{\infty}$ satisfies the Inada conditions, by the choice of $x>0$, $y$ in \eqref{9223} can take any value in $(0,\infty)$, we deduce that \eqref{convuv} holds. 
\end{proof}

\begin{proof}[Proof of Proposition \ref{lemConvOpt}]
First, for every $(z_n, n)$, $n\in\bN^*$, the existence and uniqueness of optimizers to \eqref{primalProblem} and \eqref{dualProblem} follows from Proposition \ref{propconvuv} and \cite[Theorem 3.2]{MostovyiNec}.
Next, for strictly positive $x$ and $y$, let us denote
\begin{equation}\label{defGamma}\bs
\Gamma(x,y):= & \frac 12 \( V(x) + V(  y)\)-V\( \frac {x+y}2\)   \\
= &\int_0^\infty  \frac 12   \left\{\left(   V'\(  z + \frac{y-x}2\)  -V'(  z) \)1_{\{y>x\}}(z)1_{\[x, \frac{x+y}2\]}(z)\right. \\
 &+\left. \(V'\(  z + \frac{x-y}2\)-V'(  z)  \)1_{\{x>y\}}(z)1_{\[y, \frac{x+y}2\]}(z)\right\} dz, 
 \end{split}
\ee
then, by the  convexity of $V$, $\Gamma\geq 0$, and, for $x\neq y$, by the strict convexity of $V$, $\Gamma>0$. Next, let us define 
 \be\label{defzdelta}\zeta^\delta : = 1\wedge\frac \delta 2 \inf\limits_{z\leq \frac 1\delta + \frac\delta 2}\( V'\(z + \frac \delta 2\) -V'(z) \),\quad \delta >0,
 \ee
 Then, from the {\it strict} monotonicity of $V'$ and the Inada conditions, one can show that $\zeta^\delta>0$, $\Pas$,  for every $\delta>0$.  
 
 Let $\e>0$ be fixed and let  $  Y^m\in\cY\(\frac{\e}{2^m}\)$ be such that $  Y^m$ satisfies Assumption \ref{asUI}, $m\in\bN$, and we also set $  Y^\infty= 0$ and recall from Assumption \ref{asZconv} that $Z^\infty_T=1$. 
 Next, let us define 
 \be\label{defhnm}
  h^{n, m}: = \widehat Y^n_T(z_n) +\frac{  Y^m_T}{Z^n_T}, \quad (n, m)\in(\bN^*)^2,
  \ee
 where, we recall that $\widehat Y^n_T(z_n)$ are the (dual) minimizers to \eqref{dualProblem} corresponding to $n$ and $z_n$; that is, in \eqref{dualProblem}, we have
 $$ v_n(z_n) = \bE_{\bP^n}\[V\(\widehat Y^n_T(z_n)\)\],\quad n\in\bN^*.$$ 
 
Assume by contradiction that $\widehat Y^n_T(z_n)$ does not converge in probability to $\widehat Y^\infty_T(z_\infty)$. 
Then, in view of Assumption \ref{asZconv}, there exists $\delta >0$, such that
$$\limsup\limits_{n\to\infty}\bP\[ \left| \widehat Y^n_T(z_n) - \frac{\widehat Y^\infty_T(z_\infty)}{Z^n_T} \right|>\delta\]>\delta,$$
which, in view of the construction in \eqref{defhnm}, we can rewrite as
\be\label{1221}
\limsup\limits_{n\to\infty}\bP\[ \left|h^{n, m} -\frac{h^{\infty, m}}{Z^n_T} \right|>\delta\]>\delta,\quad for~every\quad m\in\bN.
\ee
Here we stress that 
$h^{n, m} -\frac{h^{\infty, m}}{Z^n_T}$
does not depend on $m\in\bN^*$. Next, 
 by passing, if necessary, to a smaller $\delta>0$, from \eqref{1221} and Lemma \ref{lemStructure}, we get 
\be\label{10231}
\limsup\limits_{n\to\infty}\bP\[\left|h^{n, m} -h^{\infty, m} \frac{1}{Z^n_T} \right|>\delta,~~ \widehat Y^n_T(z_n) +\frac{\widehat Y^\infty_T(z_\infty)}{Z^n_T}+ 2\sum\limits_{m\in\bN}  \frac{Y^m_T}{Z^n_T} <\frac 1\delta\]>\delta.
\ee
Let us denote  
\be\label{defAn}
A_n:= \left\{\left|h^{n, m} -h^{\infty, m} \frac{1}{Z^n_T} \right|>\delta,~~ \widehat Y^n_T(z_n) +\frac{\widehat Y^\infty_T(z_\infty)}{Z^n_T}+ 2\sum\limits_{m\in\bN}  \frac{Y^m_T}{Z^n_T} <\frac 1\delta\right\}, ~~ n\in\bN.
\ee
Then, we can rewrite \eqref{10231} as
\be\label{12213}
\limsup\limits_{n\to\infty}\bP[A_n]> \delta>0.
\ee

Using Assumption \ref{asZconv} and since $\zeta^\delta$, defined in \eqref{defzdelta}, takes values in $(0,1]$, from   \eqref{12213}, we get\footnote{The short proof of \eqref{12231} is given in this paragraph: if, by contradiction, we suppose that 
 \be\label{12221}
 \limsup\limits_{n\to\infty} \bE[ \zeta^\delta Z^n_T1_{A_n}]=0,
 \ee
 then, in particular, we have 
 \be\label{12251}
 \bP\text-\lim\limits_{n\to\infty} \zeta^\delta Z^n_T1_{A_n}= 0.
 \ee 
 As $\zeta^\delta Z^n_T$, $n\in\bN$, are strictly positive and converge to $\zeta^\delta$ in probability $\bP$, we have  $$\bP\text-\lim\limits_{n\to\infty}\frac 1{\zeta^\delta Z^n_T} =  \frac 1{\zeta^\delta}>0.$$ Therefore, we deduce from \eqref{12251} that $1_{A_n}$, $n\in\bN$, converge to $0$ in probability $\bP$, and, by the Dominated Convergence Theorem, also in $\bL^1(\bP)$. This contradicts to \eqref{12213}.
 } 
 \be\label{12231}
 \limsup\limits_{n\to\infty} \bE[ \zeta^\delta Z^n_T1_{A_n}]>0.
 \ee
 Recalling that $A_n$'s are defined in \eqref{defAn}, $\zeta^\delta$ - in \eqref{defzdelta}, and $\Gamma$'s in \eqref{defGamma},  one can see that, on $A_n$, $\Gamma\(h^{n,m}, h^{\infty, m}\frac{1}{Z^n_T}\) \geq \zeta^\delta$, $m\in\bN$. Therefore, we obtain
 \be\label{12232}
 \limsup\limits_{n\to\infty} \bE_{\bP^n}\[ \Gamma\(h^{n,m}, h^{\infty, m}\frac{1}{Z^n_T}\)\] \geq \limsup\limits_{n\to\infty} \bE[ \zeta^\delta Z^n_T1_{A_n}]>0,\quad m\in\bN, \ee
 where, in the last inequality, we have used \eqref{12231}. As $ \limsup\limits_{n\to\infty} \bE[ \zeta^\delta Z^n_T1_{A_n}]$ does not depend on $m\in\bN$, we further obtain from \eqref{12232} that
  \be\label{12233}
  \limsup\limits_{m\to\infty} \limsup\limits_{n\to\infty} \bE_{\bP^n}\[ \Gamma\(h^{n,m}, h^{\infty, m}\frac{1}{Z^n_T}\)\] >0, \ee

Next, from the definition of $\Gamma$ in \eqref{defGamma}, we get 
\be\nn\bs
\frac 12\(V\(h^{n,m}\) + V\(h^{\infty, m} \frac{1}{Z^n_T}\) \)  
= V\(\frac {h^{n,m} +  h^{\infty,m} \frac{1}{Z^n_T}}{2}\)
+ \Gamma \(h^{n,m}, h^{\infty,m} \frac{1}{Z^n_T}\),\\ (n,m)\in\bN^2,~~  \Pas.
\end{split}
\ee
Taking the expectations under the respective $\bP^n$'s, we obtain
\be\label{12210}\bs
& \frac 12\(\bE_{\bP^n}\[V\(h^{n, m}\)\] + \bE_{\bP^n}\[V\(h^{\infty, m} \frac{1}{Z^n_T}\)\] \)\\  
&\geq  \bE_{\bP^n}\[ V\(\frac {h^{n,m} +  h^{\infty, m} \frac{1}{Z^n_T}}{2}\)\]  +\bE_{\bP^n}\[ \Gamma \(h^{n,m}, h^{\infty,m} \frac{1}{Z^n_T}\)\], \quad (n,m)\in\bN^2.
\end{split}\ee

Next, for every $n\in\bN^*$, since via Lemma \ref{lemStructure} we have that $h^{\infty,m} \frac{1}{Z^n_T}$ is a terminal value of an element of $\cY^n\(z_\infty + \frac{\e}{2^m}\)$ that is feasible element for \eqref{dualProblem} at $\(n,z_\infty+\frac{\e}{2^m}\)$, we can bound the right-hand side of \eqref{12210} from below by $v_n\(\frac{z^n + z_\infty}{2}+\frac{\e}{2^m}\)+\bE_{\bP^n}\[\Gamma \(h^{n,m}, h^{\infty,m} \frac{1}{Z^n_T}\) \] $, to obtain
\be\label{1222}
\bs
 &\frac 12\(\bE_{\bP^n}\[V\(h^{n, m}  \)\] + \bE_{\bP^n}\[V\(h^{\infty, m} \frac{1}{Z^n_T}\)\] \)   \\   
&\geq  v_n\(\frac{z^n + z_\infty}{2} + \frac{\e}{2^m}\) + \bE_{\bP^n}\[\Gamma \(h^{n,m}, h^{\infty,m} \frac{1}{Z^n_T}\)\] ,~~  (n,m)\in\bN^2.
\end{split}\ee
From Proposition \ref{propconvuv} (see \eqref{convuv}) and \cite[Theorem 10.8]{Rok}, we get
\be\label{1225}
\lim\limits_{n\to\infty}v_n(z_n) = v_\infty(z_\infty).
\ee
Let us fix $m\in\bN$. Then, from \eqref{1225}, 
 \eqref{defhnm}, and the monotonicity of $V$,  
 we obtain 
\be\label{1226}\bs
&v_\infty(z_\infty) = \lim\limits_{n\to\infty}v_n(z_n) = \lim\limits_{n\to\infty}\bE_{\bP^n}\[ V\(\widehat Y^n_T(z_n)\)\]  \\ \geq&\limsup\limits_{n\to\infty}\bE_{\bP^n}\[ V\(\widehat Y^n_T(z_n) + \frac{Y^m_T}{Z^n_T}\)\]=\limsup\limits_{n\to\infty}\bE_{\bP^n}\[ V\(h^{n,m}\)\].
\end{split}
\ee
We note that the upper bound in \eqref{1226}, $v_\infty(z_\infty)$, does not depend on $m\in\bN$ and that \eqref{1226} gives an asymptotic bound for the first term in \eqref{1222}, $\limsup\limits_{n\to\infty}\bE_{\bP^n}\[V\(h^{n, m}  \)\]$.

For the same fixed $m\in\bN$, let us establish an upper asymptotic bound for the second term in \eqref{1222}, $\limsup\limits_{n\to\infty}\bE_{\bP^n}\[V\(h^{\infty, m} \frac{1}{Z^n_T}\)\]$. Using Assumption \ref{asUI} and Fatou's lemma, similarly to the proof of Lemma \ref{lemDualBound}, we get
 \be\bs\label{10254}
\limsup\limits_{n\to\infty}\bE_{\bP^n}\[V\(h^{\infty, m} \frac{1}{Z^n_T}\)\] &\leq \bE\[V\(h^{\infty, m}\)\]   \\
= \bE\[V\(\widehat Y_T^{\infty}(z_\infty) + Y^m_T\)\]&\leq \bE\[V\(\widehat Y_T^{\infty}(z_\infty)\)\]= v_\infty(z_\infty),
\end{split}
\ee 
where, in the last inequality, we also used the monotonicity of $V$. Let us observe that the terminal upper bound in \eqref{10254}, $v_\infty(z_\infty)$, does not depend on $m\in\bN$. 

From \eqref{1222}, \eqref{1226}, and \eqref{10254}, we get 
\be\label{1227}
\bs
&v_\infty(z_\infty) \geq \frac 12\limsup\limits_{n\to\infty} \(\bE_{\bP^n}\[V\(h^{n, m}  \)\] + \bE_{\bP^n}\[V\(h^{\infty, m} \frac{1}{Z^n_T}\)\]   \)  \\   
&\geq  \limsup\limits_{n\to\infty} \left(v_n\(\frac{z^n + z_\infty}{2} + \frac{\e}{2^m}\) +\bE_{\bP^n}\[ \Gamma \(h^{n,m}, h^{\infty,m} \frac{1}{Z^n_T}\)\]\right) \\
&= v_\infty\(z_\infty + \frac{\e}{2^m}\) + \limsup\limits_{n\to\infty}\bE_{\bP^n}\[\Gamma \(h^{n,m}, h^{\infty,m} \frac{1}{Z^n_T}\) \] ,\quad   m \in\bN.
\end{split}\ee
Therefore,  rearranging terms in  in \eqref{1227}, taking $\limsup\limits_{m\to\infty}$, and using \eqref{12233}, we conclude that  
\be\nn
\bs 0<&
\limsup\limits_{m\to\infty}\limsup\limits_{n\to\infty} \bE_{\bP^n}\[\Gamma \(h^{n,m}, h^{\infty,m} \frac{1}{Z^n_T}\) \] \leq \limsup\limits_{m\to\infty}\(v_\infty(z) - v_\infty\(z_\infty + \frac{\e}{2^m}\)\),\end{split}
\ee
which, however, contradicts the continuity of $v_\infty$ at $z_\infty$ that follows from the finiteness of $v_\infty$, established in Proposition \ref{propconvuv}, and convexity of $v_\infty$.
%
We obtain that 
 \be\label{6291}
 \bP\text-\lim\limits_{n\to\infty}\widehat Y^n_T(z_n) =\widehat Y^\infty_T(z_\infty).
 \ee
 
 Now, let $y_n = u'_n(z_n)$, $n\in\bN^*$. We note that $y_n$'s are well-defined by the differentiability of $u_n$'s, which follows from Proposition \ref{propconvuv}, item $(1)$, and \cite[Theorem 3.2]{MostovyiNec}. Then, by Lemma \ref{propconvuv}, item $(2)$, and \cite[Theorem 25.7]{Rok}, we deduce that the sequence $y_n$, $n\in\bN$, converges to $y_\infty$. Replacing $z_n$'s with $y_n$'s in the argument above, similarly to \eqref{6291}, we get
 \be\label{6292}
 \bP\text-\lim\limits_{n\to\infty}\widehat Y^n_T(y_n) =\widehat Y^\infty_T(y_\infty).
 \ee
Since, by \cite[Theorem 3.2]{MostovyiNec} we have  $$\widehat X^n(z_n) = -V'(\widehat Y^n_T(y_n)),\quad n\in\bN^*,$$
we deduce from the  continuity of $-V'(\omega, \cdot)$, $\omega\in\Omega$, and \eqref{6292} that 
 $$\bP\text-\lim\limits_{n\to\infty}\widehat X^n_T(z_n) =\widehat X^\infty_T(z_\infty).$$
\end{proof}
\begin{proof}[Proof of Proposition \ref{propUBPconv}]
Let us observe that by Lemma \ref{lemConvOpt}, the sequence $y_n$, $n\in\bN$, converges to $y_\infty$, and the terminal values of the dual minimizers converge in probability, that is 
$$\bP\text - \lim\limits_{n\to\infty}\widehat Y^n_T(y_n) = \widehat Y^\infty_T(y_\infty).$$
This, together with Assumption \ref{asZconv}, implies that 
\be\label{1021}
 \bP\text - \lim\limits_{n\to\infty}Z^n_T\frac{\widehat Y^n_T(y_n)}{y_n} = Z^\infty_T\frac{\widehat Y^\infty_T(y_\infty)}{y_\infty}. 
\ee
Now, from the martingale property of the dual minimizers, following  \cite[Theorem 4.2]{MostovyiPietroLMP} 
one can show that  the representation \eqref{ubpRep} holds.

Next, applying the martingale property of dual minimizers again,  we deduce that  
$$1 = \lim\limits_{n\to\infty}\bE_{\bP^n}\[\frac{\widehat Y^n_T(y_n)}{y_n}\] = \bE\[\frac{\widehat Y^\infty_T(y_\infty)}{y_\infty}\],$$
and therefore, in view of \eqref{1021} and Scheffe's lemma, we deduce that the sequence $Z^n_T\frac{\widehat Y^n_T(y_n)}{y_n}$, $n\in\bN$, is uniformly integrable under $\bP$ and the convergence in \eqref{1021} also holds in $\bL^1(\bP)$. Therefore, for every bounded contingent claim $f$, the sequence $Z^n_T\frac{\widehat Y^n_T(y_n)}{y_n}f$, $n\in\bN$, is also uniformly integrable under $\bP$ and since, by \eqref{1021}, it converges to $Z^\infty_T\frac{\widehat Y^\infty_T(y_\infty)}{y_\infty}f$ in probability $\bP$, the convergence also takes place in  $\bL^1(\bP)$. Together with earlier established \eqref{ubpRep}, this implies \eqref{UBPconv}. 

\end{proof}

\section
{Indifference price invariance}\label{secInv}
In this section, we take a closer look at the contingent claims whose indifference price does not depend on the choice of the probability measure. 
\subsection{Complete Models}  Let work under $\bP$, and denote by $\cM$ the set of equivalent to $\bP$ separating measures for $S$, and suppose that 
\be\label{noArb2}
\cM\neq \emptyset,
\ee
which is equivalent to assuming \eqref{noArb} and additionally requiring that $\cY(1)$ contains a strictly positive martingale.  
 Here we show that in complete models, the indifference price of every bounded contingent claim is unique and coincides with the arbitrage-free price.  In fact, the converse is also true. This is the subject of the following proposition. 
 
\begin{prop}\label{propComplete}
Let us suppose that $S$ satisfies \eqref{noArb2} and consider a probability measure $\bQ\sim\bP$. Then, the following conditions are equivalent:
\begin{enumerate}[(i)]
\item the model is complete,
\item for every bounded contingent claim $f$, the indifference price depends neither on 
$x>0$, $U\in\cF\cV(\bQ)$,  $\bQ$, nor on $\tilde x>0$, $\tilde U\in\cF\cV(\bP)$, $\bP$, and we have
\be\label{eqIPI}
\Pi(f, x, U, \bQ) = \bE_{\widehat\bQ}\[f\] = \Pi(f, \tilde x, \tilde U, \bP),\ee
where $\widehat\bQ$ is the unique separating measure for $S$ and $\Pi$ is defined in Definition \ref{defIndPrice}. 
\end{enumerate}
\end{prop}
\begin{proof}Let us introduce 
$$\cD^{\bQ}:= \left\{\eta\in\bL^0_+(\Omega, \cF, \bP):~\eta\leq \frac{d\tilde \bQ}{d\bQ},~for~some~\tilde \bQ\in\cM \right\}.$$
$(i)\Rightarrow (ii).$ 
One can see that $\cD^{\bQ}$ is a convex and solid hull of the set of the Radon-Nykodim derivatives of the elements of $\cM$ with respect to $\bP$. It follows from the argument in \cite{KS}; see the proofs of Lemmas 4.1 and 4.2, that $\cD^{\bQ}$ is also closed in $\bL^0(\Omega, \cF, \bP)$. We observe that the completeness of the model implies that $\cD^{\bQ}$ has a unique $\Pas$ maximal element. Let us denote this element $\widehat\eta$. As $\cM\neq\emptyset$, it follows that $\bE_{\bQ}[\widehat\eta] = 1$ and that $\widehat\eta>0$, $\bQ\text-$a.s.. 
Therefore, we can define a probability measure $\widehat\bQ\sim\bP$ via its Radon-Nykodim derivative that satisfies   
$\frac{d\widehat\bQ}{d\bQ} = \widehat\eta$. One can see that $\widehat\bQ$ is a separating measure for $S$ that is unique by the completeness assumption.  

Now, let us fix $f\in\bL^\infty(\Omega, \cF, \bP)$. Then, for every $y>0$ and $U\in\cF\cV(\bQ)$, where $\cF\cV(\bQ)$ is defined in Definition \ref{defFV}, by the maximality of $\widehat\eta$ in $\cD^{\bQ}$, we obtain that $y\widehat\eta$ is the dual minimizer at $y$. As $\bE_{\bQ}\[ \widehat\eta\]  = 1$, we deduce that the indifference price of $f$ at $x>0$ associated with $U\in\cF\cV(\bQ)$ and $\bQ$ is given by   
$$\Pi(f, x, U, \bQ) = \bE_{\widehat\bQ}\[ f\],\quad x>0, \quad U\in\cF\cV(\bQ).$$
As a similar argument can be carried out under the probability measure $\bP$, we conclude that $(ii)$ holds. 

$(ii) \Rightarrow (i).$ Let $(ii)$ holds. Assume by contradiction that the model is incomplete. Let $\tilde V$ be a bounded from the above deterministic function on $[0,\infty)$, such that $-\tilde V$ satisfies Assumption \ref{Inada}, and let $\tilde U$ its convex conjugate. As $\tilde U$ is non-random and  $\tilde V$ is bounded from above, one can show that $\tilde U\in\cF\cV(\bQ)\cap\cF\cV(\bP)$. Let us work under the probability measure $\bQ$ and consider the utility maximization problem and its dual under $\bQ$, that is 
\be\label{1301}
\sup\limits_{X\in\cX(x)}\bE_{\bQ}\[ \tilde U(X_T)\],\quad x>0,
\ee
and
\be\label{1302}
\inf\limits_{Y\in\cY^{\bQ}(y)}\bE_{\bQ}\[ \tilde V(Y_T)\],\quad y>0.
\ee
As $\tilde U\in\cF\cV(\bQ)\cap\cF\cV(\bP)$, from \cite[Theorem 3.2]{MostovyiNec}, we deduce that there exists a unique solution to \eqref{1302} at $y=1$, $\widehat Y^{\bQ}\in\cY^{\bQ}(1)$. By the incompleteness of the model assumption, there exists an element $Y\in\cY^{\bQ}(1)$ such that 
$A: =\left\{ Y_T > \widehat Y^{\bQ}_T\right\}$ has $\bQ\[A\]>0$. Let us set 
$$
\alpha := \frac 12 \frac{\bE_{\bQ}\[ \(\tilde V\(\widehat Y^\bQ_T\) - \tilde V\(  Y_T\)\)1_A\]}
{\bE_{\bQ}\[ \(\tilde V\(  Y_T\) - \tilde V\(\widehat Y^\bQ_T\) \)1_{A^c}\]}.
$$
As both the numerator and denominator are {\it strictly} positive and finite-valued (by the boundedness of $\tilde V$), we have that $\alpha>0$. Next, we define
$$U(\omega, x) := \tilde U(x) 1_A + \alpha \tilde U(x)1_{A^c},\quad (\omega, x)\in\Omega\times [0,\infty).$$
Then, $U$ satisfies Assumption \ref{Inada}, and its convex conjugate $V$ is given by 
$$V(\omega, y) = \tilde V(y) 1_A + \alpha \tilde V(y)1_{A^c},\quad (\omega, y)\in\Omega\times [0,\infty).$$
Therefore, we have 
\be\bs\nn
\bE_{\bQ}\left[ V\(\widehat Y^{\bQ}_T\)\right] =& \bE_{\bQ}\left[ \(\tilde V\(\widehat Y^{\bQ}_T\) - \tilde V(Y_T)\)1_A\right]  \\
&+ 
\alpha\bE_{\bQ}\left[ \(\tilde V\(\widehat Y^{\bQ}_T\) - \tilde V(Y_T)\)1_{A^c}\right]
+\bE_{\bQ}\left[\tilde V(Y_T) \right]\\
 =&\bE_{\bQ}\left[\tilde V(Y_T) \right] + \tfrac 12\bE_{\bQ}\left[ \(\tilde V\(\widehat Y^{\bQ}_T\) - \tilde V(Y_T)\)1_A\right]>\bE_{\bQ}\left[\tilde V(Y_T) \right].
\end{split}\ee
where, in the inequality, we have used the definition of $\alpha$. 
Therefore, $\widehat Y^\bQ$ is not a minimizer  to the dual problem with the value function $V$ under the probability measure $\bQ$  at $y=1$, that is to 
\be\label{6261}
\inf\limits_{\tilde Y\in \cY^{\bQ} (1)}\bE_{\bQ}\[ V(\tilde Y_T)\].
\ee
The construction of $V$ implies that $V\in \cF\cV(\bQ)$. Therefore, by \cite[Theorem 3.2]{MostovyiNec}, there exists a unique minimizer to \eqref{6261}. Let us denote it by $\bar Y$. Then, from the respective optimality of $\bar Y$ and $\widehat Y^{\bQ}$ to minimization problems associated with the value functions $V$ and $\tilde V$, accordingly, we have
$$\bQ\[ \widehat Y^{\bQ}_T > \bar Y_T\] >0  \quad {\rm and}\quad \bQ\[\bar Y_T > \widehat Y^{\bQ}_T\] >0.$$ 
Let us consider a bounded contingent claim $f$ given by 
$$f : = \min\(1, -\tilde V'(\widehat Y^{\bQ}_T), -V'(\bar Y_T)\)1_{\{\widehat Y_T>\bar Y_T\}}.$$
Under the measure $\bQ$, let us consider the primal problems with the utilities $\tilde U$ and $U$, respectively. The construction of $\tilde U$ and $U$ implies that \eqref{finValueU} holds for both primal value functions. Then, with conjugates stochastic fields $\tilde V$ and $V$, let us consider the dual problems and denote the dual value functions by $\tilde v$ and $v$, respectively. One can see that \eqref{finValueV} holds for both dual value functions. In view of \eqref{finValueU}  and \eqref{finValueV}, the results of \cite[Theorem 3.2]{MostovyiNec} apply, and therefore $\tilde v$ and $v$ are differentiable. and the initial wealth $-\tilde v'(1)$ and $-  v'(1)$. 
Then, $0\leq f\leq -\tilde V'(\widehat Y^{\bQ}_T)$ and  $0\leq f\leq -V'(\bar Y_T)$, where $-\tilde V'(\widehat Y^{\bQ}_T)$ and $-V'(\bar Y_T)$ are the terminal values of the optimal wealth processes, which are maximal in $\cX(-\tilde v'(1))$ and $\cX(- v'(1))$, respectively. Therefore, 
one can show that the indifference prices of $f$ associated with $\tilde U$ and $U$ are given by 
$$\Pi(f, -\tilde v'(1), \tilde U, \bQ) = \bE_{\bQ}\[ \widehat Y^{\bQ}_Tf\] \quad {\rm and}\quad \Pi(f, -  v'(1),   U, \bQ)=\bE_{\bQ}\[ \bar Y_Tf\].$$
  Furthermore, from the construction of $f$, we get
$$0<\bE_{\bQ}\[(\widehat Y^{\bQ}_T - \bar Y_T)f\] =  \Pi(f, -\tilde v'(1), \tilde U, \bQ) - \Pi(f, -  v'(1),   U, \bQ).$$
Therefore, the indifference price invariance fails in the sense that \eqref{eqIPI} does not hold. This is a contradiction. We conclude that the model is complete. 
\end{proof}

\subsection{Connection to the results in \cite{Fabrice03}}

A large class of examples where our results apply and are relevant is given by the financial markets for weakly informed insiders, as in \cite{Fabrice03}. We first recall the framework of \cite{Fabrice03}.

Let $T > 0$ be a constant finite time horizon. In this section, we work on a continuous-time, arbitrage-free financial market. Namely, let $(\Omega, (\mathcal{F}_t)_{0 \le t \le T} , \mathbb P)$ be a filtered probability space that satisfies the usual conditions (i.e., the filtration $\mathcal{F}$ is complete and right-continuous) and such that, moreover, the filtration $\mathcal{F}$ is quasi-left-continuous. We assume that there are $d$ basic tradable assets, the price process of which is am $\mathcal{F}$-adapted positive local martingale $(S_t)_{0 \le t \le T}$ . In addition, we assume that $S$ is square-integrable and that $\mathcal F_0$ is trivial, which implies that $S_0$ is constant.
We assume that $(\Omega, (\mathcal{F}_t)_{0 \le t \le T} , \mathbb P)$ is complete in the sense that $(S_t)_{0 \le t \le T}$ enjoys the predictable representation property.

Let us consider an $\mathcal{F}_T$-measurable random variable $\sf Y$ 
 taking values in a Polish space $\cP$, where $\cP$ is endowed with its Borel $\sigma$-algebra $\cB(\cP)$. Given a probability measure $\nu$ on $\(\cP, \cB(\cP)\)$, we consider the  set $\mathcal{E}^\nu$ of probability measures on $(\Omega,\mathcal{F}_T)$ such that:
 \begin{enumerate}
\item $\mathbb{Q}$ is equivalent to $\mathbb{P}$;
\item The distribution of $\sf Y$ under $\mathbb Q$ is $\nu$.
\end{enumerate}
We assume that $\mathcal{E}^\nu$ is not empty, which reduces to the assumption that $\nu$ is equivalent to the distribution of $\sf Y$ under $\mathbb P$.
We recall that the space $\mathcal{A}_F(S)$  of admissible strategies is the space of $\mathbb{R}^d$-valued and $\mathcal{F}$-predictable processes $\Theta$ integrable with respect to the price process $S$ such that $  \Theta  \cdot  S $ is a $(\mathbb P, \mathcal F)$ martingale.

Let us consider a deterministic Inada utility function $U$, that is, the one satisfying Assumption \ref{Inada} and which does not depend on $\omega$. The value of the information $(\sf Y,\nu)$ for an insider  with initial wealth $x>0$ was  defined in \cite{Fabrice03} as:
\begin{align*}
u(x,\nu)=\inf_{\mathbb{Q} \in \mathcal{E}^\nu} \sup_{\Theta \in \mathcal{A}_F(S)} \mathbb{E}^{\mathbb Q} \left[ U \left(x+  \Theta  \cdot  S_T\right)\right],
\end{align*}
that is $u(x,\nu)$ is the minimal gain in the utility associated with the anticipation $(\sf Y, \nu)$. 

We have the following result.

\begin{prop}
Let $\nu_n$, $n\in\bN$, be a sequence of probability measures on $\(\cP, \cB(\cP)\)$ equivalent to $\nu$ and such that 
$$\nu\text-\lim\limits_{n\to\infty} \frac{d\nu_n}{d\nu}  = 1.$$
Let $\mathbb P^{n} := \frac{d\nu_n}{d\nu} (Y) \cdot \mathbb{P}^{\nu}$, $n\in\bN$, $\bP^\infty := \bP^\nu$, the assumptions of Proposition  \ref{propconvuv} hold, and the financial market above is complete. Then, for every $x >0$, we have 
$$
\lim_{n \to  \infty} u(x,\nu_n)=u(x,\nu).
$$
Therefore, the value of the weak information in the sense of \cite{Fabrice03} is continuous in the topology of the total variation distance.
\end{prop}

\begin{proof}
It follows from Proposition 1 and Theorem 1 in \cite{Fabrice04} that
$$
u(x,\nu_n)=\sup_{\Theta \in \mathcal{A}_F(S)} \mathbb{E}^{\mathbb P^{n}} \left[ U \left(x+  \Theta  \cdot  S_T\right)\right].
$$
 Now, the assertions follow from Proposition \ref{propconvuv}.
\end{proof}

\subsection{Incomplete models} We begin with the definition of replicability. 
\begin{defn}\label{defRepl} 
A bounded random variable $f$ is {\it replicable} if there exists an admissible wealth process $X$, such that $-X$ is also admissible and  $X_T = f$. 
\end{defn}
Next, we show that in incomplete markets, bounded {\it replicable} contingent claims are indifference price invariant.

\begin{lem}
\label{lemRepl}
Let us suppose that a separating measure exists for $S$. Then, for every $\bQ\sim\bP$, $x>0$, and $U\in\cF\cV(\bQ)$, the indifference price for every bounded replicable contingent claim $f$ satisfies
$$\Pi(f, x, U, \bQ) = x_0,$$
where $x_0$ is the initial value of the replicating strategy for $f$.
\end{lem}
\begin{proof}
Let us fix a bounded and replicable at the initial price $x_0$ contingent claim $f$. 
Recalling the definition of sets $\cX(x,q)$'s in \eqref{defXxq}, let us set
$$u(x,q) := \sup\limits_{X\in \cX(x,q)}\bE_{\bQ}\[U(X_T + qf) \], \quad (x,q)\in\bR^2,$$
where we use the convention as in \eqref{convu}, and here, if for some $(x,q)\in\bR^2$, $\cX(x,q)=\emptyset$, we set $u(x,q): = -\infty$. 

Next, for every $\tilde x\neq x_0$, let us consider $q = sign(x_0 - \tilde x)$, and we have 
\be\label{1304}
u(x - q\tilde x, q) = u(x - q\tilde x + qx_0, 0) > u(x,0), \quad x>0,\ee
as $u(\cdot, 0)$ is {\it strictly} increasing by \cite[Theorem 3.2]{MostovyiNec}. We deduce from Definition \ref{defIndPrice} 
 and \eqref{1304}  that $\tilde x$ is not an indifference price.
In turn, for $\tilde x = x_0$, we have
$$u(x - q\tilde x, q) = u(x - q\tilde x + x_0 q, 0) = u(x,0),\quad q\in\bR,$$
and so, by Definition \ref{defIndPrice}, $x_0$ is an indifference price for $f$. 
\end{proof}

%
%


The following example gives an incomplete model, where the set of contingent claims, whose indifference price does not depend on the reference probability measure, is exactly the set of replicable contingent claims. 
\begin{exa}\label{exTrin}
Let us consider a one-period trinomial model on a probability space, where $\Omega = \{\omega_i\}_{i = 1}^3$, $\cF_0$ is trivial, and $\cF_1$ is the discrete $\sigma$-field on $\Omega$. Let us suppose that there is a risky asset 
$S$ that satisfies $S_0 = 1$ and $S_1(\omega_i) = s_i$, $i = 1, 2, 3$, where $s_1>s_2>s_3>0$ are constants. 
Let us suppose that there is also a riskless security whose price is equal to $1$ at both times $0$ and $1$. Supposing that $\bP(\omega_i)>0$, $i = 1,2,3$, results in a very simple example of an {\it incomplete} model. 

Let us fix any utility function or stochastic field satisfying Assumption \ref{Inada}, and let its convex conjugate be denoted by $V$. Let us introduce the constraints that every absolutely continuous martingale measure for $S$ satisfies. 
\be\label{249}\bs
&\sum\limits_{i = 1}^3\bQ(\omega_i) = 1, \quad \bQ(\omega_j)\geq 0, j= 1,2,3, \quad {\rm and}\quad \bE_{\bQ}\[S_1\]  = S_0 = 1.
\end{split}\ee
Then, for every $y>0$, the dual minimization problem amounts to solving the following constrained minimization problem over vectors describing probability measures $\bQ$ as in \eqref{249}.
\be\label{241}
\bs
 &\min
\bE\[V\(y\frac{d\bQ}{d\bP}\)\], \\
 &{\rm subject~to~}\eqref{249}. 
\end{split}\ee
Then, one can show that the minimizer exists and is unique for every $y>0$. Let us denote it by $\widehat\bQ(y)$, $y>0$. 

Next, let us consider a different probability measure $\tilde \bP\sim\bP$ and possibly a different utility stochastic field satisfying Assumption \ref{Inada}, whose convex conjugate is denoted by $\tilde V$. Then, for every $y>0$, the dual problem analogous to \eqref{241}  can be written as 
\be\label{242}
\bs
&\min
\tilde\bE\[\tilde V\(y\frac{d\bQ}{d\tilde \bP}\)\], \\
&{\rm subject~to~} \eqref{249}.\\ 
\end{split}\ee
Let us denote the unique minimizer to \eqref{242} by $\tilde \bQ (y)$, $y>0$. 

Let us consider a contingent claim $f$ on this probability space. The indifference price invariance amounts to verifying whether the following equality holds or not.
\be\label{243}
\bE_{\hat\bQ(y_1)}\[f\] = \bE_{\tilde \bQ(y_2)}\[f\],
\ee
where $y_1$ and $y_2$ are positive constants, which are the derivatives of the corresponding primal value functions at some $x_1>0$ and $x_2>0$. Let us suppose that these $x_i$'s and, therefore, $y_i$'s are fixed and drop $y_i$'s in \eqref{243} for simplicity of notations. One can see that the constraints \eqref{249} for  \eqref{241} and \eqref{242} lead to 
$$\widehat\bQ(\omega_3) = \frac{s_2 - 1 + \widehat\bQ(\omega_1)\(s_1 - s_2\)}{s_2 - s_3}\quad {\rm and} \quad 
\widehat\bQ(\omega_2) = \frac{s_3 - 1 + \widehat\bQ(\omega_1)\(s_1 - s_3\)}{s_3 - s_2},
$$
as well as to 
$$\tilde \bQ(\omega_3) = \frac{s_2 - 1 + \tilde \bQ(\omega_1)\(s_1 - s_2\)}{s_2 - s_3}\quad {\rm and} \quad 
\tilde \bQ(\omega_2) = \frac{s_3 - 1 + \tilde \bQ(\omega_1)\(s_1 - s_3\)}{s_3 - s_2}.
$$
Therefore, \eqref{243} leads to the following description of the indifference price invariant $f = f(\omega_i)$, $i = 1,2,3$. 
\be\label{245}
(s_3 - s_2)f(\omega_1) + (s_1 - s_3)f(\omega_2) + (s_2-s_1)f(\omega_3) = 0,
\ee
which are represented by the set of vectors in $\bR^3$ that are orthogonal to $(s_3 - s_2, s_1 - s_3, s_2-s_1)^\top$. In turn, the replicable claims in this model are the ones given by 
\be\label{246}
f = \alpha S_1 + \beta,\quad (\alpha,\beta)\in\bR^2.
\ee
 One can see that \eqref{245} and \eqref{246} specify the same class of random variables. Thus, the class of contingent claims whose indifference price depends neither on the physical probability measure nor the initial wealth nor the utility stochastic field in the class $\cF\cV$ (under the associated physical measure) is precisely the class of replicable claims.
%

\end{exa}

In the previous example, indifference price invariant claims are exactly the replicable ones. In general, the class of indifference price invariant (with respect to changes in the physical probability measure) claims can be strictly bigger than the class of replicable ones. 

The next example gives an incomplete model, where {\it every} bounded contingent claim is indifference price invariant if we restrict ourselves to {\it deterministic} utilities,  with respect to a particular class of changes of the physical probability measure $\bP$.

\begin{exa}\label{exSD}
Let us suppose that $B$ and W are independent Brownian motions on the complete stochastic basis $\(\Omega, \cF, (\cF_t)_{t\in[0,T]}, \bP\)$, where the   filtration $(\cF_t)_{t\in[0,T] }$ is generated by $B$ and $W$ and made right-continuous and complete, $T\in(0,\infty)$ is a time horizon. Let us suppose that there is a riskless asset with a price process equal to $1$ at all times, and there is one risky traded asset whose price process is given by 
$$\frac{dS_t}{S_t} = \mu dt + \sigma dW,\quad t\in(0,T], \quad S_0 = 1,$$
where $\mu$ and $\sigma>0$ are constants. Let us introduce the minimal martingale measure 
$\bQ$, whose Radon-Nikodym derivative with respect to $\bP$ is given by 
\be\label{defQ}
\frac{d\bQ}{d\bP} = \cE\(-\frac{\mu}{\sigma}W\)_T.
\ee
Let us consider a deterministic Inada utility function (that is, the one satisfying Assumption \ref{Inada} and which does not depend on $\omega$), whose convex conjugate $V$ satisfies 
\be\label{2411}
\bE\[V\(y\frac{d\bQ}{d\bP}\)\]<\infty, \quad y>0.
\ee
Following 
\cite[Example 4.4]{MostovyiSirbuThaleia}, one can show that the density of $\bQ$ is the dual minimizer for every $y>0$.
 
Next, let us change the probability measure $\bP$ to $\tilde \bP$ in a way that the dynamics of $S$ under $\tilde \bP$ is given by 
$$\frac{dS_t}{S_t} = \tilde\mu dt + \sigma dW^{\tilde \bP},\quad t\in(0,T], \quad S_0 = 1,$$
where $\tilde \mu\in\bR$ and $ W^{\tilde \bP}_t = W_t + \frac{\mu - \tilde \mu}{\sigma} t$, $t\in [0,T]$, is a Brownian motion under $\tilde \bP$.  Then, for every {\it deterministic} Inada utility function $\tilde U$, whose convex  conjugate $\tilde V$ satisfies
\be\label{2412}
\tilde\bE\[\tilde V\(y\frac{d\bQ}{d\tilde\bP}\)\]<\infty, \quad y>0,
\ee
one can show that the density of $\bQ$ (defined via \eqref{defQ}) with respect to $\tilde \bP$ is the dual minimizer. 

Now, for every $f\in\bL^\infty(\Omega, \cF, \bP)$, replicable or not, where an example of a bounded nonreplicable contingent claim is $f = 1_{\{B_T>0\}}$, we have
\be\label{2415}
\Pi\(f, x, U, \bP\) = \bE_{\bQ}\[f\] = \Pi(f, \tilde x, \tilde U, \tilde \bP), 
\ee
where $U$ and $\tilde U$ are arbitrary deterministic Inada utility functions, whose convex conjugates $V$ and $\tilde V$  satisfy \eqref{2411} and \eqref{2412}, respectively, $x$ and $\tilde x$ are positive constants, $\bQ$ defined via \eqref{defQ}, and $\bP$ and $\tilde \bP$ are the probability measures specified in this example above. We note that \eqref{2411} and \eqref{2412} are adaptations of condition \eqref{finValue} to the present setting of two probability measures in the case of deterministic utilities $U$ and $\tilde U$.

We conclude that every $f\in\bL^\infty(\Omega, \cF, \bP)$ is indifference price invariant, in the sense \eqref{2415}, with respect to a change of probability measure from $\bP$ to $\tilde\bP$ that is specified in this example above.   
\end{exa}
\section{Counterexamples}\label{secCounter}
\begin{exa}
This example shows that without Assumption \ref{asUI1}, stability may fail, in particular, the assertions of Proposition \ref{propconvuv}  might not necessarily hold. 

Let us fix a probability space $(\Omega, \cF, \bP)$, where the filtration $(\cF_t)_{t\in[0,1]}$ is the usual augmentation of the filtration generated by  a Brownian motion $W$ and $1$ is the time horizon. Let us suppose that there are two traded assets one riskless with price process equal to $1$ at all times and one risky, whose dynamics is given by 
$$\frac {dS_t}{S_t} = \mu dt + \sigma d W_t, \quad t\in(0,T],\quad S_0 = 1,$$
where $\mu$ and $\sigma>0$ are constants. Let us set  
\be\label{251}
\phi := \sum\limits_{n = 1}^\infty\frac 1{2^n} \sqrt{\frac 2 n}\exp\({\(\frac 12 - \frac 1n\)W^2_1}\).
\ee
Then, one can see that $\phi\in\bL^1(\bP)$. Now, let us consider a utility stochastic field $$U(\omega, x) = \tilde U(x) - \phi(\omega),\quad (\omega, x)\in\Omega\times[0,\infty),
$$
where $\tilde U$ is an Inada deterministic utility function. We  suppose that $\tilde U$ is negative-valued, e.g., $\tilde U(x) = \frac {x^p}{p}$, $p<0$. One can see that the base model satisfies \eqref{finValueU} and \eqref{finValueV}, as the convex conjugate of $U$, $V$, is non-positive-valued. 

Next, let us consider perturbations of the probability measure $\bP$, such that the Radon-Nikodym derivatives of the new probability measures with respect to $\bP$ are given by $$Z^n_1 = 
\sqrt{ \tfrac n{n+2}}\exp\(\tfrac 1{n+2}W^2_1\), \quad n\in\bN.$$
Then, one can see that $\bE[Z^n_1]=1$, $n\in\bN$, and $\bP{\text-}\lim\limits_{n\to\infty}Z^n_1 = 1$, and so Assumption \ref{asZconv} holds.  

On the other hand, by direct computations, we have
$$\bE[Z^n_1 \phi] 
 = \infty,\quad n\in\bN,$$
and so, for every $n\in\bN$ and every  wealth process $X\in\bigcup\limits_{x\geq 0}\cX(x)$, we have
\be\label{252}\bs
&\bE\[ Z^n_TU^-(X_T)\] = -\bE\[Z^n_T U(X_T)\] \\&= \bE\[ Z^n_T\(\tilde U(X_T) + \phi\)\]\geq \bE\[ Z^n_T  \phi\] = \infty.
\end{split}
\ee
Therefore, {\it Assumption \ref{asUI1} does not hold}
and, in view of the convention \eqref{convu}, from \eqref{252}, we have  \be\label{2421}u_n(x) = -\infty, \quad x>0,\quad n\in\bN.\ee
On the other hand, as $\phi\in\bL^1(\bP)$, we have that $u_\infty(x)\in(-\infty, 0)$, $x>0$, which together with \eqref{2421} implies that the convergence of the primal value functions in the sense of Proposition \ref{propconvuv}, item $(2)$, fails. Moreover, \eqref{2421} shows that the finiteness of $u_n$, $n\in\bN$, as in Proposition \ref{propconvuv}, item $(1)$, fails too. 
\end{exa}
As pointed out in Remark \ref{remDetU}, Assumption \ref{asUI1} holds under Assumption \ref{asZconv} if $U$ is deterministic. The next example shows that without Assumption \ref{asUI}, the assertions of Proposition \ref{propconvuv} may fail.

\begin{exa}
Let us consider the probability space and the financial model consisting of two assets, as in the previous example. Let us suppose that the preferences of an economic agent are modeled by a deterministic utility function $U(x) = \frac {x^p}p$, $x>0$, where $p\in\(\tfrac 12,1\)$. Then, the convex conjugate of $U$, $V$, is given by 
$V(y) = \frac {y^{-q}}q$, $y>0$, where $q = \frac{p}{1 - p}>1.$ Under the probability measure $\bP$, the unique minimizer to the dual problem, for every $y>0$, is given by $y\frac{d\bQ}{d\bP}$, where $$\frac{d\bQ}{d\bP} = \cE\(-\frac\mu\sigma W\)_1.$$
In particular, the dual value function for the base model is finite-valued, as $\exp\(q \frac \mu\sigma W_1\)\in\bL^1(\bP)$, that is 
\be\label{254}
v_\infty(y)<\infty,\quad y>0.
\ee
Now, with 
$$c_n := \bE\[\exp\(-\frac 1n W^3_11_{\{W_1\geq 0\}}\)\], \quad n\in\bN,$$
let us consider  $$Z^n_1 = \frac 1{c_n}\exp\(-\frac 1n W^3_11_{\{W_1\geq 0\}}\), \quad n\in\bN.$$
Then, one can see that Assumption \ref{asZconv} holds. 
However, we get 
\be\label{256}\bs
&\bE\[ Z^n_1 V^+\(\frac{y \frac {d\bQ}{d\bP}}{Z^n_1}\)\] = \bE\[ Z^n_1 V\(\frac{y \frac {d\bQ}{d\bP}}{Z^n_1}\)\] = 
\bE\[ Z^n_1 \frac{1}{q}\(\frac{y \frac {d\bQ}{d\bP}}{Z^n_1}\)^{-q}\]\\
&=\frac{\exp\(\frac {\mu^2 q}{2\sigma^2}\)y^{-q}}q \bE\[ \(Z^n_1\)^{1-q} \exp\(q\frac \mu\sigma W_1\)\] 
 \\
&= \frac{\exp\(\frac {\mu^2 q}{2\sigma^2}\)y^{-q}}{qc^{1-q}_n} \bE\[  \exp\(\frac {q-1}n W^3_11_{\{W_1\geq 0\}} + q\frac \mu\sigma W_1\)\] = \infty.
\end{split}\ee
This implies that Assumption \ref{asUI} does not hold. 
Therefore, in view of the convention \eqref{convv}, using \eqref{256}, we deduce that \be\label{255}
v_n(y) = \infty, \quad y>0, \quad n\in\bN,
\ee
which implies that the finiteness of the dual value functions, as in Proposition \ref{propconvuv}, item $(1)$, does not hold. Furthermore, from \eqref{254} and \eqref{255}, we conclude that the  convergence of the dual value functions in Proposition \ref{propconvuv}, item $(2)$, also does not hold. 
\end{exa}
\bibliographystyle{alpha}\bibliography{literature}
\end{document}